\documentclass[11pt,reqno]{amsart}
\usepackage{amssymb}
\usepackage{amsmath}
\usepackage{amsmath}
\usepackage{amsthm}
\usepackage{amsfonts}
\usepackage{graphicx}
\usepackage{enumerate}
\usepackage[bottom,flushmargin,hang]{footmisc}
\usepackage{parskip} 
\usepackage{bm}
\usepackage{bbm} 
\usepackage{dsfont}
\usepackage[authoryear,longnamesfirst]{natbib}
\usepackage{multirow}

\usepackage{multicol}
\usepackage{tikz}

\numberwithin{equation}{section}
\newtheorem{theorem}{Theorem}

\newtheorem{corollary}{Corollary}

\newtheorem{lemma}{Lemma}
\theoremstyle{definition}

\newtheorem{remark}{Remark}

\newcommand{\supp}{\mathrm{supp}}
\renewcommand{\tilde}{\widetilde}

\newcommand{\Gn}{\mathbb{G}_{n}}

\newcommand{\calF}[0]{\mathcal{F}}
\newcommand{\calG}[0]{\mathcal{G}}

\newcommand{\calS}[0]{\mathcal{S}}

\newcommand{\E}[0]{\mathbb{E}}
\newcommand{\R}[0]{\mathbb{R}}

\newcommand{\N}{\mathbb{N}}

\renewcommand{\i}{{\bm{i}}}
\newcommand{\bN}{{\bm{N}}}
\newcommand{\EN}{\mathbb{\E}_{N}}
\newcommand{\e}{{\bm{e}}}
\newcommand{\calE}{\mathcal{E}}

\begin{document}

\title[Maximal Inequalities for SE Processes]{Maximal Inequalities for Separately Exchangeable Empirical Processes}
\author[H. D. Chiang]{Harold D. Chiang}

\address[Harold D. Chiang]{Department of Economics, University of Wisconsin-Madison, 1180 Observatory Drive Madison, WI 53706-1393, USA.}
\email{hdchiang@wisc.edu}
\thanks{First draft: 14 February 2025. I wish to express my sincere gratitude to Bruce E. Hansen for his encouragement and comments, which has been instrumental in the preparation of this manuscript.
}
\begin{abstract}
This paper derives new maximal inequalities for empirical processes associated with separately exchangeable random arrays. For fixed index dimension $K\ge 1$, we establish a global maximal inequality bounding the $q$-th moment ($q\in[1,\infty)$) of the supremum of these processes. We also obtain a refined local maximal inequality controlling the first absolute moment of the supremum. Both results are proved for a general pointwise measurable function class. Our approach uses a new technique partitioning the index set into transversal groups, decoupling dependencies and enabling more sophisticated higher moment bounds.

\end{abstract}
 \allowdisplaybreaks
\maketitle
\section{Introduction}

This paper develops novel local and global maximal inequalities for empirical processes of separately exchangeable arrays, where the index dimension \(K\in\mathbb{N}\) is fixed but arbitrary and the empirical processes are defined on general classes of functions. 
Separately exchangeable (SE) arrays are widely utilised in modelling multiway-clustered random variables and/or \(K\)-partite networks in econometrics and statistics (see, for example, \citealt{davezies2021empirical,mackinnon2021wild,menzel2021bootstrap,chiang2022multiway,graham2024sparse}). 

Maximal inequalities are powerful tools that are indispensable in the analysis of numerous econometric and statistical problems. They have proven crucial in areas such as semiparametric estimation and debiased machine learning (e.g. \citealt{belloni2015uniform,chernozhukov2018double}), quantile and instrumental variable quantile regression (see, for instance, \citealt{kato2012asymptotics,chetverikov2016iv,galvao2020unbiased}), conditional mode estimation (\citealt{ota2019quantile}), testing many moment inequalities \citep{chernozhukov2019inference}, adversarial learning (\citealt{kaji2023adversarial}), targeted minimum loss-based estimation (\citealt{van2017generally}), and density estimation for dyadic data (\citealt{cattaneo2024uniform}), to name but a few.

Despite promising recent advances in the studies of SE arrays and its potential benefits for statistical inference and econometric applications, the theoretical literature on maximal inequalities for SE arrays remains rather scarce. In particular, no general global maximal inequality is available. While certain global maximal inequalities have been established for special cases—for instance, Theorem B.2 in \cite{chiang2023inference} provides a bound for a setting with general \(K\) and any \(q\)-th moment (with \(q\in[1,\infty)\)) albeit only for a finite class of functions, and Lemma C.3 in \cite{liu2024estimation} extends the result for a general class of functions but only for the first absolute moment (\(q=1\)) in two-way (\(K=2\)) settings
—no general global maximal inequality applicable to arbitrary \(K,q\) and an infinite pointwise measurable class of functions has been developed for SE arrays. 

The local maximal inequality is a crucial tool for obtaining sharp convergence rates.
Apart from the classical i.i.d. case (\(K=1\)), no local maximal inequality for SE arrays exists in the literature.
 Establishing a local maximal inequality for SE arrays is especially challenging because their intricate multiway dependence structure induces complex interactions among observations. Unlike in the i.i.d. or \(U\)-statistics settings, the multidimensional dependencies inherent to SE arrays render classical techniques—such as symmetrisation and Hoeffding averaging—inapplicable.
To address these challenges, we propose a novel proof strategy that carefully partitions the index set into transversal groups. This innovative construction effectively decouples the dependencies among observations within each group, thereby enabling the application of the Hoffmann–J\o rgensen inequality to derive sharp bounds for terms involving higher moments. Consequently, our work fills a critical gap by providing both global and local maximal inequalities for a potentially uncountable but pointwise measurable class of functions under SE sampling.

These methodological advances are underpinned by foundational results from empirical process theory. For textbook treatments, see e.g. \cite{vdVW1996,delaPenaGine1999,gine2016mathematical}. In particular, the local maximal inequalities for i.i.d.\ random variables established in \cite{vanderVaartWellner2011} and \cite{CCK2014AoS}, as well as maximal inequalities for \(U\)-statistics/processes presented in \cite{Chen2018,ChenKato2019b}, provide the technical backbone for our arguments. Our results also built directly upon the symmetrisation and Hoeffding type decomposition for SE arrays developed in \cite{chiang2023inference}.

We follow the fundamental notation for SE arrays as presented in \cite{davezies2021empirical, chiang2023inference}. To fix ideas, let $K$ be a fixed positive integer and denote by
\[
\bm{i} = (i_1, i_2, \dots, i_K) \in \mathbb{N}^K,
\]
a $K$-tuple index. Given a probability space $(S,\mathcal{S},P)$, suppose that \[\{ X_{\bm{i}} : \bm{i} \in \mathbb{N}^K \}\] 
is a collection of $\mathcal{S}$-valued random variables satisfying the separate exchangeability (SE) and dissociation (D) conditions defined below.
\begin{enumerate}
\item[(SE)]  For any $\pi =(\pi_1,\dots,\pi_K)$, a $K$-tuple of permutations of $\N$, $\{X_{\bm{i}}\}_{\i\in \N^K}$ and $\{X_{\pi(\bm{i})}\}_{\i\in \N^K}$ are identically distributed.
\item[(D)]  For any two set of indices $I,I'\subset \N^K$, $\{X_{\bm{i}}\}_{\i\in I}$ and $\{X_{\bm{i}}\}_{\i\in I'}$  are independent.
\end{enumerate}
Under Conditions (SE) and (D), the Aldous--Hoover--Kallenberg (AHK) representation (see Corollary 7.35 in \citealt{kallenberg2005probabilistic}) guarantees the existence of the following representation:
\begin{align}
	X_{\bm{i}} = \tau\Bigl( \{ U_{\bm{i} \odot \bm{e}} \}_{\bm{e} \in \{0,1\}^K \setminus \{\bm{0}\}} \Bigr), \label{eq:AHK_representation}
\end{align}
where $\odot$ denotes the Hadamard (element-wise) product, the collection 
\[
\{ U_{\bm{i} \odot \bm{e}} : \bm{i} \in \mathbb{N}^K,\; \bm{e} \in \{0,1\}^K \setminus \{\bm{0}\} \}
\]
consists of mutually independent and identically distributed (i.i.d.) random variables, and $\tau$ is a Borel measurable map taking values in $\calS$.

	Let $\bm{N} = (N_1, N_2, \dots, N_K)$ and define 
	\[
	[\bm{N}] = \prod_{k=1}^{K} \{1,2,\dots,N_k\}.
	\]
	Also, denote
	\[
	N = \prod_{k=1}^K N_k,\quad n = \min\{N_1, N_2, \dots, N_K\},\quad \text{and}\quad \overline{N} = \max\{N_1, N_2, \dots, N_K\}.
	\]
	We say a class of functions $\calF:\calS\to \R$ is pointwise measurable if there exists a countable subclass $\calF'\subset \calF$ such that for each $f\in \calF$, there exists a sequence $(f_j)_j\subset \calF'$ such that $f_j\to f$ pointwisely.
	Given the observed set of random variables  \(\{X_{\bm{i}}:\i \in [\bm{N}]\}\) that satisfy Conditions (SE) and (D), and a pointwise measurable class of functions $\calF$ with elements $f : \calS \to \mathbb{R}$, define the sample mean process by
	\[
	\EN f = \frac{1}{N} \sum_{\bm{i} \in [\bm{N}]} f(X_{\bm{i}})
	\]
	and the empirical process by
	\[
	\Gn(f) = \frac{\sqrt{n}}{N} \sum_{\bm{i} \in [\bm{N}]} \Bigl\{ f(X_{\bm{i}}) - \E\bigl[f(X_{\bm{1}})\bigr] \Bigr\},
	\]
	where $\bm{1}=(1,...,1)$.
	Without loss of generality, assume that $\E[f(X_{\bm{1}})] = 0$ for all $f \in \calF$. In this paper, we establish inequalities that control the $q$-th moment of the supremum of the empirical process,
	\(
 \E\bigl[\|\Gn\|_\calF^q \bigr] ,
	\)
	for some $q \in [1,\infty)$.

\subsection{Notation}
Let $\N$ denote the set of positive integers and $\R$ for the real line. For $a,b\in \R$, let $a\vee b=\max\{a,b\}$ and $a\wedge b = \min\{a,b\}$. 
 Denote for $m\in \mathbb N$ that
$
[m] = \{1,2,\ldots,n\}.
$ 
For two real vectors $\bm{a}= (a_{1},\dots,a_{K})$ and $\bm{b} = (b_{1},\dots,b_{K})$, we denote $\bm{a} \le \bm{b}$ for $a_{j} \le b_{j}$ for all $1 \le j \le K$. 
Let $\supp(\bm{a}) = \{ j : a_j \ne 0\}$. 
We denote by $\odot$ the Hadamard product, i.e., for $\i = (i_1,\dots,i_K)$ and $\bm{j} = (j_1,\dots,j_K)$, $\i \odot \bm{j}= (i_1 j_1,\dots,i_K j_K)$. 
For each $k=1,2,...,K$, define $\calE_k=\{\e\in \{0,1\}^K: \e \odot (1,...,1)=k\}$ and thus $\{0,1\}^K=\cup_{k=1}^K \calE_k$. For $q\in [1,\infty]$, let $\|f\|_{Q,q}=(Q|f|^q)^{1/q}$. For a non-empty set $T$ and $f:T\to \R$, denote $\|f\|_T=\sup_{t\in T}|f(t)|$. For a pseudometric space $(T,d)$, let $N(T,d,\varepsilon)$ denote the $\varepsilon$-covering number for $(T,d)$. We say $F:\calS\to \R_+$ is an envelope for a class of functions $\calF\ni f:\calS\to \R$ if $\sup_{f\in\calF}|f(x)|\le F(x)$ for all $x\in\calS$. For $0 < \beta < \infty$, let $\psi_{\beta}$ be the function on $[0,\infty)$ defined by $\psi_{\beta} (x) = e^{x^{\beta}}-1$. Let $\| \cdot \|_{\psi_{\beta}}$ denote the associated Orlicz norm, i.e., 
$\| \xi \|_{\psi_\beta}=\inf \{ C>0: \E[ \psi_{\beta}( | \xi | /C)] \leq 1\}$ for a real-valued random variable $\xi$.

\section{Main Results}\label{sec:maximal_inequalities_SE}

Before presenting the main results, let us first introduce the Hoeffding type decomposition from \cite{chiang2023inference}.
	 For any \(\bm{i}\in [\bm{N}]\), define
	\begin{align*}
		(P_{\bm{e}}f)\Bigl(\{U_{\bm{i}\odot \bm{e}'}\}_{\bm{e}'\le \bm{e}}\Bigr)
		&=\E\Bigl[f(X_{\bm{i}})\,\Big|\,\{U_{\bm{i}\odot \bm{e}'}\}_{\bm{e}'\le \bm{e}}\Bigr]. 
	\end{align*}
	We then define recursively for \(k=1,2,\dots, K\) that
	\begin{align*}
		(\pi_{\bm{e}_k}f)(U_{\bm{i}\odot \bm{e}_k})
		&=(P_{\bm{e}_k}f)(U_{\bm{i}\odot \bm{e}_k}),
	\end{align*}
	and for \(\bm{e}\in \bigcup_{k=2}^K\calE_k\) set
	\begin{align*}
		(\pi_{\bm{e}}f)\Bigl(\{U_{\bm{i}\odot \bm{e}'}\}_{\bm{e}'\le \bm{e}}\Bigr)
		=\;& (P_{\bm{e}}f)\Bigl(\{U_{\bm{i}\odot \bm{e}'}\}_{\bm{e}'\le \bm{e}}\Bigr) \nonumber\\
		&\quad-\sum_{\substack{\bm{e}'\le \bm{e}\\ \bm{e}'\ne \bm{e}}}
		(\pi_{\bm{e}'}f)\Bigl(\{U_{\bm{i}\odot \bm{e}''}\}_{\bm{e}''\le \bm{e}'}\Bigr). 
	\end{align*}
	Note that by the AHK representation \eqref{eq:AHK_representation}, for a fixed \(\bm{e}\) the distributions of
	\[
	(P_{\bm{e}}f)\Bigl(\{U_{\bm{i}\odot \bm{e}'}\}_{\bm{e}'\le \bm{e}}\Bigr)
	\quad\text{and}\quad
	(\pi_{\bm{e}}f)\Bigl(\{U_{\bm{i}\odot \bm{e}'}\}_{\bm{e}'\le \bm{e}}\Bigr)
	\]
	do not depend on the index \(\bm{i}\). Hence, we shall write \(P_{\bm{e}}f\) and \(\pi_{\bm{e}}f\) for a generic \(\bm{i}\).
	
	Now, fix any \(1\le k\le K\) and let \(\bm{e}\in \calE_k\). Then, by Lemma~1 in \cite{chiang2023inference}, for any \(\ell\in \supp(\bm{e})\) the random variable
	\(
	(\pi_{\bm{e}}f)\Bigl(\{U_{\bm{i}\odot \bm{e}'}\}_{\bm{e}'\le \bm{e}}\Bigr)
	\)
	has mean zero conditionally on \(\{U_{\bm{i}\odot \bm{e}'}\}_{\bm{e}'\le \bm{e}-\bm{e}_{\ell}}\).
	In addition, define
	\(
	I_{\bm{N},\bm{e}} = \{\bm{i}\odot \bm{e} : \bm{i}\in [\bm{N}]\}.
	\)
	Then, we have
	\[
	\bigl|I_{\bm{N},\bm{e}}\bigr| = \prod_{k'\in\supp(\bm{e})} N_{k'}.
	\]
	Accordingly, define
	\begin{align*}
		H_{\bm{N}}^{\bm{e}}(f)
		=\frac{1}{\bigl|I_{\bm{N},\bm{e}}\bigr|} \sum_{\bm{i}\in I_{\bm{N},\bm{e}}}
		(\pi_{\bm{e}}f)\Bigl(\{U_{\bm{i}\odot \bm{e}'}\}_{\bm{e}'\le \bm{e}}\Bigr),
	\end{align*}
	 we now obtain the Hoeffding-type decomposition
	\begin{align*}
		\EN f = \sum_{k=1}^K \sum_{\bm{e}\in\calE_k} H_{\bm{N}}^{\bm{e}}(f).
	\end{align*}
	To bound \(\E\bigl[ \|\Gn(f)\|_{\calF} \bigr] \), it thus suffices to control each individual term \(\E\bigl[\|H_{\bm{N}}^{\bm{e}}(f)\|_{\calF}\bigr]\) separately.
	
	Finally, fix any \(1\le k\le K\) and \(\bm{e}\in \calE_k\).  For a given class $\calF$ with an envelope $F$, define for a $\delta>0 $ its \emph{uniform entropy integral} by
	\begin{align*}
		J_{\bm{e}}(\delta) = J_{\bm{e}}(\delta,\calF,F)
		:= \int_{0}^{\delta} \sup_{Q} \Biggl\{ 1+ \log N\Bigl(P_{\bm{e}}\calF,\|\cdot\|_{Q,2},\tau\|P_{\bm{e}}F\|_{Q,2}\Bigr) \Biggr\}^{k/2} d\tau, 
	\end{align*}
	where
	\[
	P_{\bm{e}}\calF := \{P_{\bm{e}}f : f\in \calF\},
	\]
	and the supremum is taken over all finite discrete distributions \(Q\).
	
The following result is a general global maximal inequality for SE empirical processes with an arbitrary index order $K$ and for a general order of moment $q\in[1,\infty)$. Its proof follows the  arguments in the proof of Corollary B.1 in \cite{chiang2023inference} with some modifications  to account for a  more general class of functions.
\begin{theorem}[Global maximal inequality for SE processes]\label{theorem:global_maximal_ineq}
	Suppose $\calF:\calS\to \R$ is a pointwise measurable class of functions with an envelope $F$.
	Let $(X_\i)_{\i\in [\bm{N}]}$ be a sample from $S$-valued separately exchangeable random vectors $(X_\i)_{\i\in \N^K}$. Pick any $1 \le k \le K$ and $\bm{e} \in \mathcal{E}_{k}$. Then, for any $q \in [1,\infty)$, we have
	\[
	|I_{\bN,\e}|^{1/2}\left (\E \left [ \left \| H_{\bN}^\e(f)  \right \|_{\calF}^{q} \right ] \right)^{1/q} \lesssim
	J_\e(1)  \| F \|_{P,q\vee 2}.
	\]
\end{theorem}
\begin{proof}
	
	By symmetrisation inequality for SE processes (Lemma B.1  in \cite{chiang2023inference}; note that it is dimension free), for independent Rademacher r.v.'s $(\varepsilon_{1,i_1})$,...,$(\varepsilon_{k,i_k})$ that are independent of $(X_\i)_{\i\in\N^K}$, one has
	\begin{align*}
		|I_{\bN,\e}|^{1/2}	\left(\E[\|H_\bN^\e(f)\|_\calF^q]\right)^{1/q}
	=&\left (\E \left [ \left \|\frac{1}{\sqrt{|I_{\bN,\e}|}} \sum_{\i\in I_{\bN,\e}} (\pi_{ \e}f )(\{U_{\i \odot \e'}\}_{\e'\le \e})   \right \|_{\calF}^{q} \right ] \right)^{1/q}\\
		\lesssim& 
		\left (\E \left [ \left \|\frac{1}{\sqrt{|I_{\bN,\e}|}} \sum_{\i\in I_{\bN,\e}}\varepsilon_{1,i_1}...\varepsilon_{k,i_k}\cdot(\pi_{ \e}f )(\{U_{\i \odot \e'}\}_{\e'\le \e}) \right \|_{\calF}^{q} \right ] \right)^{1/q}.
	\end{align*}
	By  convexity of supremum and $\cdot \mapsto (\cdot)^q$, Jensen's inequality implies that the RHS above can be upperbounded  up to a constant that depends only on $q$, $K$, and $k$ by
	\begin{align*}
		\left (\E \left [ \left \|\frac{1}{\sqrt{|I_{\bN,\e}|}} \sum_{\i\in I_{\bN,\e}}\varepsilon_{1,i_1}...\varepsilon_{k,i_k}\cdot(P_\e f )(\{U_{\i \odot \e'}\}_{\e'\le \e}) \right \|_{\calF}^{q} \right ] \right)^{1/q}.
	\end{align*}
Denote $\mathbb P_{I_{\bN,\e}}=|I_{\bN,\e}|^{-1}\sum_{\i \in I_{\bN,\e}}\delta_{\{U_{\i \odot \e'}\}_{\e'\le \e}}$, the empirical measure on the support of $\{U_{\i \odot \e'}\}_{\e'\le \e}$.
Observe that conditionally on $\{X_\i\}_{\i\in\bN}$, the object
\begin{align*}
\frac{1}{\sqrt{|I_{\bN,\e}|}}\sum_{\i\in I_{\bN,\e}}\varepsilon_{1,i_1}...\varepsilon_{k,i_k}(P_{\e}f )(\{U_{\i\odot \e'}\}_{\e'\le \e})
\end{align*}
is a homogeneous Rademacher chaos processes of order $k$.
	By Lemma \ref{lemma: Orlicz norms}, $L^q$ norm is bounded from above by $\psi_{2/k} $-norm up to a constant depends only on $(q,k)$, and thus by applying Corollary 5,1.8 in \cite{delaPenaGine1999}, one has
	\begin{align*}
		&\left (\E \left [ \left \|\frac{1}{\sqrt{|I_{\bN,\e}|}} \sum_{\i\in I_{\bN,\e}}\varepsilon_{1,i_1}...\varepsilon_{k,i_k} \cdot(P_{ \e}f )(\{U_{\i \odot \e'}\}_{\e'\le \e})  \right \|_{\calF}^{q} \right ] \right)^{1/q}\\
		\lesssim&
		\E \left [ \left\| \left \|\frac{1}{\sqrt{|I_{\bN,\e}|}} \sum_{\i\in I_{\bN,\e}}\varepsilon_{1,i_1}...\varepsilon_{k,i_k}\cdot(P_{ \e}f )(\{U_{\i \odot \e'}\}_{\e'\le \e})  \right \|_{\calF}\right\|_{\psi_{2/k}|(X_\i)_{\i\in [\bN]}}\right ] \\
		\lesssim&
		\E \left [ \int_0^{\sigma_{I_{\bN,\e}}}\left[1+\log N\left(P_{ \e}\calF,\|\cdot\|_{\mathbb P_{I_{\bN,\e}},2},\tau\right)\right]^{k/2}d \tau \right ],
	\end{align*}
	where $\sigma_{I_{\bN,\e}}^2:=\sup_{f\in\calF}\| P_\e f \|_{\mathbb P_{I_{\bN,\e}},2}^2
	$. Using a change of variable and the definition of $J_\e$, the above bound becomes
	\begin{align*}
		&\E \left [ \int_0^{\sigma_{I_{\bN,\e}}}\left[1+\log N\left(P_{ \e} \calF,\|\cdot\|_{\mathbb P_{I_{\bN,\e}},2},\tau\right)\right]^{k/2}d \tau \right ] \\
		=&\E \left [\|P_{ \e} F\|_{\mathbb P_{I_{\bN,\e}},2} \int_0^{\sigma_{I_{\bN,\e}}/\|P_{ \e} F\|_{\mathbb P_{I_{\bN,\e}},2}}\left[1+\log N\left(P_{ \e} \calF,\|\cdot\|_{\mathbb P_{I_{\bN,\e}},2},\tau\|P_{ \e} F\|_{\mathbb P_{I_{\bN,\e}},2}\right)\right]^{k/2}d \tau \right ]\\
		\le&
		\E \left [\|P_{ \e} F\|_{\mathbb P_{I_{\bN,\e}},2} J_\e\left(\sigma_{I_{\bN,\e}}/\|P_{ \e} F\|_{\mathbb P_{I_{\bN,\e}},2}\right) \right ]\\
		\le&J_\e\left(1\right) \|F\|_{P,q\vee 2} ,
	\end{align*}
	where the last inequality follows from Jensen's inequality.
\end{proof}

Although the global maximal inequality works for general $q$, in the case that the supremum of the first absolute moment is concerned, local maximal inequalities usually provides shaper bounds. 
The following is a  novel local maximal inequality for SE empirical processes. Unlike the proof of the global maximal inequality, which is largely analogous to the corresponding results for $U$-processes that can be found in \cite{ChenKato2019b}, its proof relies on a novel argument that utilises construction of a partition with the property that each block in the partition satisfies a transversality property. We present this construction in Lemma \ref{lemma:partition} below.
\begin{theorem}[Local maximal inequality for SE processes]\label{theorem:local_max_ineq}
		Suppose $\calF:\calS\to \R$ is a pointwise measurable class of functions  with an envelope $F$.
	Let $(X_\i)_{\i\in [\bm{N}]}$ be a sample from $S$-valued separately exchangeable random vectors $(X_\i)_{\i\in \N^K}$. Set $\e\in \{0,1\}^K$ and let $\sigma_\e$ be a constant such that 
	$\sup_{f\in\calF}\|P_\e f \|_{P,2}\le \sigma_\e \le \|P_\e F\|_{P,2}$,
	\[ \delta_\e=\sigma_\e/\|P_\e F\|_{P,2}\quad \text{ and } \quad M_\e=\max_{t\in [n]}(P_\e F)\left(\{U_{(t,...,t)\odot \e'}\}_{\e'\le \e}\right),\] then
	\begin{align*}
		|I_{\bN,\e}|^{1/2}\E \left [ \left \| H_{\bN}^\e(f)  \right \|_{\calF} \right ]  \lesssim
		J_\e(\delta_\e)
			\|P_\e F\|_{P,2}
	 +\frac{J_\e^2(\delta_\e)\|M_\e\|_{P,2}}{\sqrt{n}\delta_\e^2}.
	\end{align*}
\end{theorem}
\begin{remark}\label{rem:local_max_ineq}
Although our proof strategy broadly follows that of Theorem 5.1 in \cite{ChenKato2019b}, a key divergence arises. In Theorem~5.1, the Hoffmann–J\o rgensen inequality (which requires independence) is applied via the classical \(U\)-statistic technique of Hoeffding averaging (see, for example, Section~5.1.6 in \citealt{serfling1980approximation}). However, the more intricate dependence structure inherent to separately exchangeable arrays renders Hoeffding averaging inapplicable in our context. To address this challenge, we introduce an alternative approach by establishing Lemma \ref{lemma:partition}, which partitions the index set \(I_{\bN,\e}\) into \(n\) transveral groups. Together with AHK representation \eqref{eq:AHK_representation}, this yields i.i.d. elements within each group, thereby facilitating the application of the Hoffmann–J\o rgensen inequality.
\end{remark}

\begin{proof}
	
We first state a  crucial technical lemma which will be used in the following proof, a proof of this lemma is provided in the end of this section.
	\begin{lemma}[Partitioning into transversal groups]\label{lemma:partition}
		For any $\e\in \{0,1\}^K$, $I_{\bN,\e}$ can be partitioned into subsets $G$'s of size $n$ such that each $G$ is transversal, that is, any two distinct tuples
		$
		(i_1,i_2,\dots,i_K),$	$(i_1',i_2',\dots,i_K')\in G
		$
		satisfy
		\[
		i_1\ne i_1',\quad i_2\ne i_2',\quad \dots,\quad i_K\ne i_K'.
		\]
	\end{lemma}
	
	We now present the proof of Theorem \ref{theorem:local_max_ineq}.
	For an $\e\in \calE_1$, the summands are i.i.d. and thus the desired result follows directly from Lemma \ref{lemma:local_max_ineq_indep}. Therefore, we assume $K\ge 2$ and $\e\in \calE_k$ for a $k\in\{2,...,K\} $.  Assume without loss of generality that $\e$ consists of $1$'s in its first $k$ elements and zero elsewhere. By applying the symmetrisation of Lemma B.1 in \cite{chiang2023inference}, one has, for independent Rademacher r.v.'s $(\varepsilon_{1,i_1})$,...,$(\varepsilon_{k,i_k})$ that are independent of $(X_\i)_{\i\in\N^K}$, that 
	\begin{align*}
			|I_{\bN,\e}|^{1/2}\E[\|H_\bN^\e(f)\|_\calF]
		\lesssim&
		\E\left[\left|\frac{1}{\sqrt{|I_{\bN,\e}|}}\sum_{\i\in I_{\bN,\e}}\varepsilon_{1,i_1}...\varepsilon_{k,i_k}(\pi_{\e}f )(\{U_{\i\odot \e'}\}_{\e'\le \e})\right\|_\calF\right].
	\end{align*}
	Further, by convexity of supremum and Jensen's inequality, the RHS above can be upperbounded  up to a constant that depends only on $K$ and $k$  by
	\begin{align*}
			\E\left[\left|\frac{1}{\sqrt{|I_{\bN,\e}|}}\sum_{\i\in I_{\bN,\e}}\varepsilon_{1,i_1}...\varepsilon_{k,i_k}(P_{\e}f )(\{U_{\i\odot \e'}\}_{\e'\le \e})\right\|_\calF\right]
	\end{align*}
	
Denote $\mathbb P_{I_{\bN,\e}}=|I_{\bN,\e}|^{-1}\sum_{\i \in I_{\bN,\e}}\delta_{\{U_{\i \odot \e'}\}_{\e'\le \e}}$, the empirical measure on the support of $\{U_{\i \odot \e'}\}_{\e'\le \e}$.
Observe that conditionally on $\{X_\i\}_{\i\in\bN}$, the object
\begin{align*}
	R_\bN^\e(f)=\frac{1}{\sqrt{|I_{\bN,\e}|}}\sum_{\i\in I_{\bN,\e}}\varepsilon_{1,i_1}...\varepsilon_{k,i_k}(P_{\e}f )(\{U_{\i\odot \e'}\}_{\e'\le \e})
\end{align*}
is a homogeneous Rademacher chaos processes of order $k$. Further, following Corollary 3.2.6 in \cite{delaPenaGine1999}, for any $f,f'\in \calF$
\begin{align*}
\left\|R_\bN^\e(f)-R_\bN^\e(f')\right\|_{\psi_{2/k}|\{X_\i\}_{\i\in \bN}}\lesssim \left\|R_\bN^\e(f)-R_\bN^\e(f')\right\|_{\mathbb P_{I_{\bN,\e}},2}.
\end{align*}
Hence the diameter of the function class $\calF$ in $\|\cdot\|_{\psi_{2/k}|\{X_\i\}_{\i\in \bN}}$-norm is upperbounded by  $\sigma_{I_{\bN,\e}}^2$ up to a constant, where
 $\sigma_{I_{\bN,\e}}^2:=\sup_{f\in\calF}\| P_{\e}f \|_{\mathbb P_{I_{\bN,\e}},2}^2
$. 
By applying
 Fubini's theorem, Corollary 5,1.8 in \cite{delaPenaGine1999}, and a change of variables, we have
	\begin{align*}
		&\E\left[\left\|\frac{1}{\sqrt{|I_{\bN,\e}|}}\sum_{\i\in I_{\bN,\e}}\varepsilon_{1,i_1}...\varepsilon_{k,i_k}(P_{\e}f )(\{U_{\i\odot \e'}\}_{\e'\le \e})\right\|_\calF\right]\\
			\lesssim&
			\E\left[\left\|\left\|\frac{1}{\sqrt{|I_{\bN,\e}|}}\sum_{\i\in I_{\bN,\e}}\varepsilon_{1,i_1}...\varepsilon_{k,i_k}(P_{\e}f )(\{U_{\i\odot \e'}\}_{\e'\le \e})\right\|_\calF\right\|_{\psi_{2/k}|\{X_\i\}_{\i\in\bN}}\right]\\
		\lesssim&
		\E \left [ \int_0^{\sigma_{I_{\bN,\e}}}\left[1+\log N\left(P_{ \e}\calF,\|\cdot\|_{\mathbb P_{I_{\bN,\e}},2},\tau\right)\right]^{k/2}d \tau \right ] \\
		=&\E \left [
	\|P_\e F\|_{\mathbb P_{I_{\bN,\e}},2}
			 \int_0^{\sigma_{I_{\bN,\e}}/\|P_\e F\|_{\mathbb P_{I_{\bN,\e}},2}}\left[1+\log N\left(P_{ \e}\calF,\|\cdot\|_{\mathbb P_{I_{\bN,\e}},2},\tau\|P_\e F\|_{\mathbb P_{I_{\bN,\e}},2}\right)\right]^{k/2}d \tau \right ]\\
		\le&
		\E \left [\|P_\e F\|_{\mathbb P_{I_{\bN,\e}},2} J_\e\left(\sigma_{I_{\bN,\e}}/\|P_\e F\|_{\mathbb P_{I_{\bN,\e}},2}\right) \right ].
	\end{align*}
	By Lemma \ref{lemma:uniform_entropy_integral}, an application of Jensen's inequality yields
	\begin{align}
		|I_{\bN,\e}|^{1/2}\E[\|H_{\bN}^\e (f)\|_\calF]\lesssim&
		\|P_\e F\|_{P,2} J_\e\left(z\right), \label{eq:local_maximal_ineq_prelimary_bound}
	\end{align}
	where $z:=\sqrt{\E[\sigma_{I_{\bN,\e}}^2]/\|P_\e F\|_{P,2}^2}$.

	We now bound 
	\begin{align*}
		\E[\sigma_{I_{\bN,\e}}^2]=&\E\left[\left\|\frac{1}{|I_{\bN,\e}|}\sum_{\i \in I_{\bN,\e}}(P_\e f )^2\left(\{U_{\i \odot \e'}\}_{\e'\le \e}\right)\right\|_\calF\right].
	\end{align*}
We aim to apply the Hoffmann–J\o rgensen inequality to handle the squared summands. However, because the summands are not independent, we invoke Lemma \ref{lemma:partition}. By applying this lemma, we obtain a partition \(\mathcal{G}\) of \(I_{\bN,\e}\) into \(|\calG|=|I_{\bN,\e}|/n\) groups, each containing \(n\) i.i.d. observations. The i.i.d. property follows from the AHK representation \eqref{eq:AHK_representation} and the fact that within each group, any two observations share no common indices \(i_1,\dots,i_K\).

For each group \(G = \{\i_{1}(G), \i_{2}(G), \dots, \i_{n}(G)\} \in \mathcal{G}\), we define
\[
D_{f,\e}(G) = \frac{1}{n}\sum_{t=1}^n \Bigl(P_{\e}f\Bigr)^2\!\left(\{U_{\i_t(G)\odot \e'}\}_{\e'\le \e}\right),
\]
and let
\[
D_{f,\e} = \frac{1}{n}\sum_{t=1}^n \Bigl(P_{\e}f\Bigr)^2\!\left(\{U_{(t,\dots,t)\odot \e'}\}_{\e'\le \e}\right).
\]
		Then 
		we have  
		\begin{align*}
			\frac{1}{|I_{\bN,\e}|}\sum_{\i \in I_{\bN,\e}}(P_{\e}f )^2\left(\{U_{\i \odot \e'}\}_{\e'\le \e}\right)=\frac{1}{|I_{\bN,\e}|/n}\sum_{G\in \calG} D_{f,\e}(G).
		\end{align*}
Note that for each \(G\in\calG\), the AHK representation in \eqref{eq:AHK_representation} implies that \(D_{f,\e}\) and \(D_{f,\e}(G)\) are identically distributed. Consequently, by Jensen's inequality, we have
\begin{align*}
	\E\Bigl[\sigma_{I_{\bN,\e}}^2\Bigr] 
	&= \E\!\left[\left\|\frac{1}{|I_{\bN,\e}|/n}\sum_{G\in\calG} D_{f,\e}(G)\right\|_\calF\right] 
	\le \E\!\left[\left\|D_{f,\e}\right\|_\calF\right].
\end{align*}
Let us denote this bound by
\[
B_{n,\e} :=\E\!\left[\left\|D_{f,\e}\right\|_\calF\right]= \E\!\left[\left\|\frac{1}{n}\sum_{t=1}^n \Bigl(P_{\e}f\Bigr)^2\!\Bigl(\{U_{(t,\dots,t)\odot \e'}\}_{\e'\le \e}\Bigr)\right\|_\calF\right].
\]
	Thus $z\le\tilde z:=\sqrt{B_{n,\e}}/\|(\pi_\e)F\|_{P,2}$.
	Note that by symmetrisation inequality for independent processes, the contration principle (Theorem 4.12. in \citealt{ledoux1991probability}), and the Cauchy-Schwarz inequality, one has
	\begin{align*}
		B_{n,\e}=&\E\left[\left\|\frac{1}{n}\sum_{t=1}^n (P_{\e}f )^2\left(\{U_{(t,...,t)\odot \e'}\}_{\e'\le \e}\right)\right\|_{\calF}\right]\\
		\le&
		\sigma_\e^2 +\E\left[\left\|\frac{1}{n}\sum_{t=1}^n\left\{ (P_{\e}f )^2\left(\{U_{(t,...,t)\odot \e'}\}_{\e'\le \e}\right)-\E\left[(P_{\e}f )^2\right]\right\}\right\|_{\calF}\right]\\
		\lesssim&
		\sigma_\e^2 +\E\left[\left\|\frac{1}{n}\sum_{t=1}^n\varepsilon_t\cdot (P_{\e}f )^2\left(\{U_{(t,...,t)\odot \e'}\}_{\e'\le \e}\right)\right\|_{\calF}\right]\\
		\lesssim&
		\sigma_\e^2 +\E\left[M_\e\left\|\frac{1}{n}\sum_{t=1}^n\varepsilon_t\cdot (P_{\e}f )\left(\{U_{(t,...,t)\odot \e'}\}_{\e'\le \e}\right)\right\|_{\calF}\right]\\
		\le&
		\sigma_\e^2 +\|M_\e\|_{P,2}\sqrt{\E\left[\left\|\frac{1}{n}\sum_{t=1}^n\varepsilon_t\cdot (P_{\e}f )\left(\{U_{(t,...,t)\odot \e'}\}_{\e'\le \e}\right)\right\|_{\calF}^2\right]}.
	\end{align*}
	An application of Hoffmann-J\o rgensen's inequality (Proposition A.1.6 in \citealt{vdVW1996}) gives 
	\begin{align*}
		&\sqrt{\E\left[\left\|\frac{1}{n}\sum_{t=1}^n\varepsilon_t\cdot (P_{\e}f )\left(\{U_{(t,...,t)\odot \e'}\}_{\e'\le \e}\right)\right\|_{\calF}^2\right]}\\
		\lesssim&
		\E\left[\left\|\frac{1}{n}\sum_{t=1}^n\varepsilon_t\cdot (P_{\e}f )\left(\{U_{(t,...,t)\odot \e'}\}_{\e'\le \e}\right)\right\|_{\calF}\right]+\frac{1}{n}\|M_\e\|_{P,2}.
	\end{align*}
By employing analogous reasoning to that used in the initial part of the proof, we deduce that
	\begin{align*}
		&\E\left[\left\|\frac{1}{\sqrt{n}}\sum_{t=1}^n\varepsilon_t\cdot (P_{\e}f )\left(\{U_{(t,...,t)\odot \e'}\}_{\e'\le \e}\right)\right\|_{\calF}\right]\\
		\lesssim& \|P_\e F\|_{P,2}\int_{0}^{\tilde z} \sup_{Q}\sqrt{1+\log N(P_\e\calF,\|\cdot\|_{Q,2},\epsilon \|P_\e F\|_{Q,2})}d\epsilon.
	\end{align*}
	Note that the integral on the RHS can be bounded by $J_\e(\tilde z)$ and thus
	\begin{align*}
		B_{n,\e}\lesssim& \sigma_\e^2 + n^{-1}\|M_\e\|_{P,2}^2 + n^{-1/2} \|M_\e\|_{P,2} \|P_\e F\|_{P,2}J_\e(\tilde z).
	\end{align*}

	Define $$\Delta=(\sigma_\e\vee n^{-1/2}\|M_\e\|_{P,2})/\|P_\e F\|_{P,2},$$  it then follows that
	\begin{align*}
		\tilde z^2\lesssim \Delta^2 + \frac{\|M_\e\|_{P,2}}{\sqrt{n}\|P_\e F\|_{P,2}}J_\e(\tilde z).
	\end{align*}
	By applying Lemma \ref{lemma:uniform_entropy_integral} and  Lemma 2.1 of \cite{vanderVaartWellner2011} with $J=J_\e$, $A=\Delta$, $B=\sqrt{\|M_\e\|_{P,2}/\sqrt{n}\|P_\e F\|_{P,2}}$ and $r=1$, it yields that
	\begin{align*}
		J_\e(z)\le J_\e(\tilde z)\lesssim J_\e(\Delta)\left\{1+J_\e(\Delta)\frac{\|M_\e\|_{P,2}}{\sqrt{n}\|P_\e F\|_{P,2} \Delta^2}\right\}.
	\end{align*}
	Combining this with (\ref{eq:local_maximal_ineq_prelimary_bound}), we obtain the bound
	\begin{align}
		|I_{\bN,\e}|^{1/2}\E[\|H_{\bN}^\e (f)\|_\calF]\lesssim&
		J_\e(\Delta)\|P_{\e} F\|_{P,2} +\frac{J_\e^2(\Delta)\|M_\e\|_{P,2}}{\sqrt{n}\Delta^2}.\label{eq:local_maximal_ineq_main_bound}
	\end{align}
	Notice that $\delta_\e\le \Delta$ by their definitions. By Lemma \ref{lemma:uniform_entropy_integral}(iii), one has 
	\begin{align*}
		J_\e(\Delta)\le \Delta\frac{J_\e(\delta_\e)}{\delta_\e}=\max\left\{J_\e(\delta_\e), \frac{\|M_\e\|_{P,2} J_\e(\delta_\e)}{\sqrt{n}\|P_\e F\|_{P,2}\delta_\e}\right\}\le\max\left\{J_\e(\delta_\e), \frac{\|M_\e\|_{P,2} J_\e^2(\delta_\e)}{\sqrt{n}\|P_\e F\|_{P,2}\delta_\e^2}\right\},
	\end{align*}
	where the second inequality follows from the fact $J_\e(\delta_\e)/\delta_\e\ge J_\e(1)\ge 1$. Finally, using Lemma \ref{lemma:uniform_entropy_integral}(iii),
	\begin{align*}
		\frac{J_\e^2(\Delta)\|M_\e\|_{P,2}}{\sqrt{n}\Delta^2}\le \frac{J_\e^2(\delta_k)\|M_\e\|_{P,2}}{\sqrt{n}\delta_k^2}
	\end{align*}
	Combining the calculations with the bound in (\ref{eq:local_maximal_ineq_main_bound}), we have the desired inequality.
	
	\medskip
	
	\subsection*{Proof of Lemma \ref{lemma:partition}}
		For $K=1$, the result is trivial.  For $K\ge 2$,
		assume without loss of generality that
		$
		N_1 \ge N_2 \ge \cdots \ge N_K.
		$
		We prove for the case of $\e=(1,...,1)$ and $I_{\bN,\e}=[\bN]$ since other cases follow exactly the same arguments.

		Our goal is to partition $[\bN]$ into subsets (which we call \emph{groups}) of size \(N_K\) that are transversal.
		For \(j=1,2,\dots,K-1\), define $\phi_j\colon [N_K]\times [N_j] \to [N_j]$ by
		\[
		\phi_j(t,g)=  ((t+g-2) \mod N_j) + 1.
		\]
		That is, for each \(t\in [N_K]\) and \(g\in [N_j]\) the value \(\phi_j(t,g)\) is computed by adding \(t\) and \(g-1\), reducing modulo \(N_j\) (so that the result lies in \(\{0,1,\dots,N_j-1\}\)), and then adding 1 to get an element of \([N_j]\).
		For the \(K\)-th coordinate we set
		$
		\phi_K(t) = t \quad \text{for } t\in [N_K].
		$		
		Index the groups by
		\[
		(g_1,g_2,\dots,g_{K-1})\in [N_1]\times [N_2]\times \cdots \times [N_{K-1}].
		\]
		Then, for each such \((g_1,\dots,g_{K-1})\), define
		\[
		G_{(g_1,\dots,g_{K-1})} = \Bigl\{\, \Bigl( \phi_1(t,g_1),\, \phi_2(t,g_2),\, \dots,\, \phi_{K-1}(t,g_{K-1}),\, t \Bigr) : t\in [N_K] \,\Bigr\}.
		\]
		Thus, each group contains \(N_K\) elements.

		We now claim the transversality property in each group.
		For a fixed group \(G_{(g_1,\dots,g_{K-1})}\) and a fixed coordinate \(j\) (with \(1\le j\le K-1\)), the \(j\)th coordinate of an element is given by
		\[
		\phi_j(t,g_j) =  ((t+g_j-2) \mod N_j) + 1.
		\]
		Since the mapping
		\(
		t \mapsto ((t+g_j-2) \mod N_j) + 1
		\)
		is injective (note that \(N_K\le N_j\) so that there is no collision in the range), it follows that the \(j\)-th coordinates of the elements of \(G_{(g_1,\dots,g_{K-1})}\) are all distinct. For the \(K\)-th coordinate, the identity mapping \(t \mapsto t\) is trivially injective.

		Next we show the covering of $[\bN]$ and  disjointness of the groups.
		Recall that the total number of groups is 
		\(
	 N_1\cdot N_2\cdots N_{K-1}.
		\)
		Each group has \(N_K\) elements; hence, the union of all groups has
		\(
		(N_1\cdot N_2\cdots N_{K-1})\cdot N_K = N
		\)
		elements. For surjectivity, let $x = (x_1, x_2, \ldots, x_K)$ be an arbitrary element of 
		$[\bN] = [N_1] \times [N_2] \times \cdots \times [N_K]$. We wish to show that there exist 
		$(g_1, g_2, \ldots, g_{K-1}) \in [N_1] \times [N_2] \times \cdots \times [N_{K-1}]$ and 
		$t \in [N_K]$ such that 
		$$
		x = \Bigl( \phi_1(t, g_1),\, \phi_2(t, g_2),\, \dots,\, \phi_{K-1}(t, g_{K-1}),\, t \Bigr).
		$$ 
		Set $t = x_K$. Then for each $j = 1, 2, \ldots, K-1$, we must have 
		$
		\phi_j(x_K, g_j) = x_j,
		$
		where by definition $\phi_j(x_K, g_j) = (((x_K + g_j - 2) \bmod N_j) + 1)$. Notice that for 
		each fixed $x_K$, the mapping 
		$$
		g \mapsto ((x_K + g - 2) \bmod N_j) + 1
		$$ 
		is an affine function (with coefficient $1$) on the cyclic group $\mathbb{Z}/N_j$, and hence it is a bijection from $[N_j]$ onto $[N_j]$. Thus, for each $j$ there exists a unique $g_j \in [N_j]$ such that $\phi_j(x_K, g_j) = x_j$. Therefore, every $x \in S$ can be uniquely written in the form 
		$$
		\Bigl( \phi_1(x_K, g_1),\, \phi_2(x_K, g_2),\, \dots,\, \phi_{K-1}(x_K, g_{K-1}),\, x_K \Bigr),
		$$ 
		which shows that the mapping 
		$$
		\tau: (g_1,\dots,g_{K-1},t) \mapsto \Bigl( \phi_1(t, g_1),\, \phi_2(t, g_2),\, \dots,\, \phi_{K-1}(t, g_{K-1}),\, t \Bigr)
		$$ 
		is bijective.

		Thus, the collection
		\[
		\mathcal{G} = \Bigl\{\, G_{(g_1,\dots,g_{K-1})} : (g_1,\dots,g_{K-1})\in [N_1]\times \cdots \times [N_{K-1}]\,\Bigr\}
		\]
		is a partition of $[\bN]$ into groups of size \(N_K\), and in every group the entries in each coordinate are distinct.

	\end{proof}


Following Chapter 3.7 in \cite{gine2016mathematical}, a function class $\calF$ on $\calS$ with envelope $F$ is called Vapnik–Chervonenkis-type (VC-type) with characteristics $(A,v)$ if 
\begin{align*}
	\sup_Q N(\calF,\|\cdot\|_{Q,2},\varepsilon\|F\|_{Q,2})\le \left(\frac{A}{\varepsilon}\right)^v\: \text{ for all } 0<\varepsilon\le 1,
\end{align*}
where the supremum is taken over all finite discrete distributions.
By adapting the arguments used in the proofs of Corollaries 5.3 and 5.5 and Lemma 5.4 in \cite{ChenKato2019b}, we derive the following local maximal inequality for VC-type function classes.
\begin{corollary}\label{cor:local_maximal_ineq_VC}
	Under the same setting as in Theorem \ref{theorem:local_max_ineq}. In addition, suppose $\calF$ is of VC-type with characteristics $A\ge (e^{2(K-1)}/16)\vee e$ and $v\ge 1$, then for each $\e\in \calE_k$, one has
	\begin{align*}
		|I_{\bN,\e}|^{1/2}\E \left [ \left \| H_{\bN}^\e(f)  \right \|_{\calF} \right ]  \lesssim
		\sigma_\e\{v\log (A\vee \overline N)\}^{k/2} +\frac{\|M_\e\|_{P,2}}{\sqrt{n}}\{v\log (A\vee \overline N)\}^k.
	\end{align*}
\end{corollary}

\section{Conclusion}

In this paper, we have derived novel maximal inequalities for empirical processes associated with separately exchangeable (SE) arrays. Our contributions include a global maximal inequality that bounds the \(q\)-th moment of the supremum for any \(q\in[1,\infty)\), as well as a refined local maximal inequality controlling the first absolute moment, both established for a general pointwise measurable class of functions. These results extend the literature beyond the i.i.d. case and overcome the challenges posed by the intricate dependence structure of SE arrays.

A key innovation of our approach is the introduction of a new proof technique—partitioning the index set into transversal groups—which circumvents the limitations of classical tools such as Hoeffding averaging. This advancement not only fills an important gap in the theoretical framework for SE arrays, but also paves the way for more robust applications in econometric and statistical inference.
Future research may build on these findings to further explore maximal inequalities under even broader conditions and to enhance their applicability in high-dimensional and machine learning contexts.

\appendix
\section{Auxiliary Lemmas}
The following restates Theorem 5.2 in \cite{CCK2014AoS}, which is a modification of Theorem 2.1 in \cite{vanderVaartWellner2011} to allow for an unbounded envelope. 
\begin{lemma}[Local maximal inequality under i.i.d.]\label{lemma:local_max_ineq_indep}
Let $X_1,...,X_n$ be $S$-valued i.i.d. random variables. Suppose $0<\|F\|_{P,2}<\infty$ and let $\sigma^2$ be any positive constant such that $\sup_{f\in \calF} P f^2\le \sigma^2 \le \|F\|_{P,2}^2$. Set $\delta^2=\sigma/\|F\|_{P,2} $ and $B=\sqrt{\E[\max_{i\in [n]} F^2(X_i)]}$. Then
\begin{align*}
\E[\|\Gn f\|_\calF] 
\lesssim \|F\|_{P,2} J(\delta,\calF,F) + \frac{B J^2(\delta,\calF,F)}{\delta^2 \sqrt{n}}.
\end{align*}
Suppose that, in addition, $\calF$ is VC-type with characteristics $(A,v)$.
Then
\begin{align*}
\E[\|\Gn f\|_\calF] 
\lesssim \sigma\sqrt{v  \log\left(\frac{A\|F\|_{P,2}}{\sigma}\right)} + \frac{vB }{\sqrt{n}} \log\left(\frac{A\|F\|_{P,2}}{\sigma}\right).
\end{align*}
\end{lemma} 
The following restates Lemma B.3 in \cite{chiang2023inference}.
\begin{lemma}[Bounding $L^q$-norm by Orlicz norm]
	\label{lemma: Orlicz norms}
	Let $0 < \beta < \infty$ and $1 \le q < \infty$ be given, and let $m = m(\beta,q)$ be the smallest positive integer satisfying $m\beta \ge q$. Then for every real-valued random variable $\xi$, we have $(\E[|\xi|^q])^{1/q} \le (m!)^{1/( m\beta)} \| \xi \|_{\psi_\beta}$. 
\end{lemma}

The following is analogous to Lemma 5.2 in \cite{ChenKato2019b} and Lemma A.2 in \cite{CCK2014AoS}.
\begin{lemma}[Properties of $J_\e(\delta)$]\label{lemma:uniform_entropy_integral}
Suppose that $J_\e(1)<\infty$ for $\e\in\{0,1\}^K$, then for all $\e\in\{0,1\}^K$,
\begin{enumerate}[(i)]
\item $\delta \mapsto J_\e(\delta)$ is non-decreasing and concave.
\item For $c\ge1$, $J_\e(c\delta)\le c J_\e(\delta)$.
\item $\delta\mapsto J_\e(\delta)/\delta$ is non-increasing. 
\item  $(x,y)\mapsto J_\e(\sqrt{x/y})\sqrt{y}$ is jointly concave in $(x,y)\in [0,\infty)\times (0,\infty)$.
\end{enumerate}
\end{lemma}

\bibliographystyle{ecta}
\bibliography{biblio}

\end{document}